\documentclass[letterpaper, 10 pt, conference]{ieeeconf}  % Comment this line out if you need a4paper

\IEEEoverridecommandlockouts                              % This command is only needed if 
\usepackage{amsmath}
\usepackage{amssymb}
\usepackage{amsfonts}
\usepackage{graphicx}
\usepackage{enumerate}
\usepackage{mathrsfs}
\usepackage{float}
\usepackage{cite}
\usepackage{tikz}
\usetikzlibrary{matrix,arrows,decorations.pathmorphing}
\usepackage{verbatim}
\usepackage{mathtools}
\usepackage{caption}
\usepackage{subcaption}
\usetikzlibrary{arrows}

\usepackage{tabularx}

\DeclareMathOperator*{\sgn}{\mathsf{sgn}}
\DeclareMathOperator*{\Dg}{\mathsf{diag}}

\title{\LARGE \bf
Interplay Between Homophily-Based Appraisal Dynamics and Influence-Based Opinion Dynamics: Modeling and Analysis}

\author{Fangzhou Liu$^{1}$, Shaoxuan Cui$^{1}$, Wenjun Mei$^{2*}$, Florian D\"orfler$^{2}$ and Martin Buss$^{1}$% <-this % stops a space
\thanks{This work was supported by the joint Sino-German Research Project, which is funded through the German Research Foundation (DFG) and the National Science Foundation China (NSFC) 61761136005.}% <-this % stops a space
\thanks{$^{1}$F. Liu, S. Cui, and M. Buss are with the Chair of Automatic Control Engineering, Technical University of Munich, Munich, 80333 Germany {\tt\small \{fangzhou.liu, shaoxuan.cui, mb\}@tum.de}}
\thanks{$^{2}$W. Mei and F. D\"orfler are with the Automatic Control Laboratory, ETH, 8092 Zurich, Switzerland {\tt\small \{wmei, dorfler\}@ethz.ch}}%
}

\begin{document}

\maketitle
\thispagestyle{empty}
\pagestyle{empty}

\newtheorem{remark}{Remark}
\newtheorem{lemma}{Lemma}
\newtheorem{thm}{Theorem}
\newtheorem{example}{Example}
\newtheorem{definition}{Definition}
\newtheorem{prop}{Proposition}

%%%%%%%%%%%%%%%%%%%%%%%%%%%%%%%%%%%%%%%%%%%%%%%%%%%%%%%%%%%%%%%%%%%%%%%%%%%%%%%%
\begin{abstract}
 In social systems, the evolution of interpersonal appraisals and individual opinions are not independent processes but intertwine with each other. Despite extensive studies on both opinion dynamics and appraisal dynamics separately, no previous work has ever combined these two processes together. In this paper, we propose a novel and intuitive model on the interplay between homophily-based appraisal dynamics and influence-based opinion dynamics. We assume that individuals' opinions are updated via the influence network constructed from their interpersonal appraisals, which are in turn updated based on the individual opinions via the homophily mechanism. By theoretical analysis, we characterize the set of equilibria and some transient behavior of our model. Moreover, we establish the equivalence among the convergence of the appraisal network to social balance, the modulus consensus of individual opinions, and the non-vanishing appraisals. Monte Carlo validations further show that the non-vanishing appraisals condition holds for generic initial conditions. Compared with previous works that explain the emergence of social balance via person-to-person homophily mechanism, our model provides an alternative explanation in terms of the person-to-entity homophily mechanism. In addition, our model also describes how individuals' opinions on multiple irrelevant issues become correlated and converge to modulus consensus over time-varying influence networks. 
\end{abstract}
\begin{keywords}
structural balance, appraisal dynamics, opinion dynamics, network multi-agent systems
\end{keywords}

%%%%%%%%%%%%%%%%%%%%%%%%%%%%%%%%%%%%%%%%%%%%%%%%%%%%%%%%%%%%%%%%%%%%%%%%%%%%%%%%
\section{INTRODUCTION}

\paragraph{Motivation and problem description} 
In social sciences, it has been extensively studied, as opinion dynamics, how individual opinions are shaped by social influences. However, few mathematical models have been established to explain how individual opinions react to interpersonal relations. In fact, such reactions are frequently observed and even being intentionally made use of. For example, politicians and news media sometimes throw certain issues to the public to generate conflicts and divisions of our society.

In this paper, we propose a novel model on the interplay between opinion dynamics and interpersonal relations, and investigate its consequences. We assume that the interpersonal appraisals, i.e., how much people like or dislike each other, are determined by the opinion homophily mechanism. That is, individuals holding similar opinions tend to be friendly to each other, and vice versa. In the meanwhile, the interpersonal influences are proportional to the appraisals and construct a time-varying signed influence network, on which individual opinions are iterated via opinion dynamics.
%In social science, opinion dynamics study how individual opinions are shaped by interpersonal influences, while appraisal dynamics study how individuals adjust their interpersonal appraisals. These two processes are by nature intertwined since, intuitively, the interpersonal influences are determined by how the individuals appraise each other. Despite extensive studies on both opinion dynamics and appraisal dynamics separately, few efforts have been taken to combine these two processes together.

%In this paper, we propose a novel model on the interplay between the influence-based opinion dynamics and the homophily-based appraisal dynamics, and show that such interplay can lead to modulus opinion consensus and social balance, a specific configuration of the interpersonal appraisal network. Our model assumes that individuals update their opinions over the influence network constructed from their interpersonal appraisals, which in turn evolve based on the individuals' opinions via the homophily mechanism. The homophily mechanism, widely observed and supported by empirical study~\cite{mcpherson_ars_2001,fu_sr_2012}, abstracts the phenomena that people tend to befriend those who think alike. 

\paragraph{Literature Review} Opinion dynamics on signed influence networks have drawn considerable attention recently~\cite{altafini2012consensus,anton_tac_2016,JL-XC-TB-MAB:17,liu_tac_2019,GS-CA-JSB:19}. Its dynamical behavior depends on whether the signed influence network satisfies \emph{social balance}~\cite{altafini2012consensus}, i.e., whether it can be partitioned into two antagonistic factions~\cite{heider1967attitudes}, where the social links within each faction are all non-negative and the social links between the two factions are all non-positive. According to \cite{altafini2012consensus}, individual opinions evolving on a connected influence network satisfying social balance converge to \emph{bipartite consensus}. That is, individuals in one faction reach consensus on some opinion $a$, while individuals in the other faction reach consensus on $-a$. Specially, one faction can be empty, which leads to \emph{consensus} of all individual's opinions. If the signed influence network does not satisfy social balance, everyone's opinion converges to zero.

Recently, various models have been proposed to explain how the interpersonal appraisal networks evolve to social balance, e.g., see~\cite{TA-PLK-SR:05,rijt_jms_2011,marvel2011continuous,PJ-NEF-FB:13n,traag2013dynamical,mei2019dynamic}. Some of them are based on the homophily mechanism, see~\cite{traag2013dynamical,mei2019dynamic}. However, the homophily mechanism in these papers is the person-to-person homophily, i.e., individuals holding similar appraisals of the others tend to be friendly to each other. Therefore, the models in~\cite{traag2013dynamical,mei2019dynamic} are self-driven dynamics of the appraisal networks, with no opinion dynamics involved. As pointed out by Heider~\cite{heider1967attitudes}, the person-to-entity homophily could also play a role in shaping the interpersonal appraisals. This partly motivates our paper, where ``entity'' refers to individuals opinions on certain issues irrelevant to appraisals.

%The model proposed in \cite{traag2013dynamical} is proved to achieve social balance in finite-time for generic initial conditions. However, this model suffers from the divergence of interpersonal appraisals. To overcome this drawback,  Mei et. al \cite{mei2019dynamic} propose an alternative model in discrete-time manner and introduce normalization. Note that in these models, the homophily mechanism adopted is in fact the person-to-person homophily. That is, people tend to befriend those who have similar appraisals of other people. Different from these models, the homophily mechanism adopted in our model is the person-to-entity homophily, specifically, the person-to-opinion homophily, i.e., people tend to befriend those who hold similar opinions. Such person-to-entity homophily, according to Heider~\cite{heider1967attitudes}, is also considered to play a role in shaping the interpersonal appraisals.

\paragraph{Contributions} 
To the best of our knowledge, our model is the first one that studies the interplay between opinion dynamics and appraisal dynamics in social systems. Theoretical analysis shows that our proposed model is well-defined and the interpersonal appraisals and individual opinions enjoy bounded behaviors. We further characterize the set of equilibria of our model and their local stabilities. Moreover, we establish the equivalence among the convergence of the appraisal network to social balance, the modulus consensus of individual opinions, and the non-vanishing appraisals. Apart from the theoretical results, numerical study shows the validity of the non-vanishing appraisals condition for almost all generic initial conditions. In terms of sociological interpretations, our model explains the emergence of social balance via the person-to-entity homophily combined with the evolution of individual opinions. Our model also describes how individuals' opinions on irrelevant issues eventually become correlated due to the formation of antagonistic factions.

\paragraph{Organization} The remainder of this paper is organized as follows. We introduce notations and necessary definitions in Section II as well as the model description. The main theoretical analysis are presented in Section III. Finally, we conduct the numerical experiments in Section IV to validate our model and theoretical contribution.

\section{Notations, Definitions, and Model Set-Up}
\emph{Notations:} Notations frequently used are defined in Table \ref{tab:notation} and adhere closely to those in~\cite{mei2019dynamic}. In this paper, we denote by $X=(X_{ij})_{n\times n}$ \emph{appraisal matrix} for a group of $n$ individuals. Here $X_{ij}$ denotes individual $i$'s appraisal of $j$, i.e., $X_{ij}>0$ ($X_{ij}<0$ resp.) if $i$ likes (dislikes resp.) $j$. $X_{ij}=0$ is $i$ does not know $j$ or holds a neutral attitude towards $j$. The appraisal matrix $X$ defines a weighted and directed graph $G(X)$, referred to as the \emph{appraisal network}.

\begin{table}[]
\caption{Notations frequently used in this paper}
\label{tab:notation}
\begin{tabular}{@{}ll@{}}
\hline
$\mathbf{1}_{n}$ &The all ones $n \times 1$ vector  \\ 
$\mathbb{R}$ ($\mathbb{Z}_{\geq 0}$) & Set of real numbers (non-negative integers)  \\
$\mathbf{0}_{m \times n}$ & A zero matrix with $m \times n$ dimensions\\
$|X|$ & Entry-wise absolute value of matrix $X$ \\
 $\sgn(X)$ & Entry-wise sign of $X$, whose entry at $i$th row \\
& and $j$th column is $\sgn(X_{ij})$\\%, i.e., $\sgn(X_{ij}) = 1$ if$X_{ij} > 0$,\\
 %& $\sgn(X_{ij}) = -1$ if $X_{ij} < 0$, and $\sgn(X_{ij}) = 0$ \\
 %& if $X_{ij} = 0$.  \\
$|X|_{max} $ $ (|X|_{min})$ & $\max_{i,j}|X_{ij}|$ ($\min_{i,j}|X_{ij}|$) \\
$X_{i*}$ ($X_{*i}$) & The ith row (column) vector of $X$  \\
$A<B$ & $A_{ij}<B_{ij}$ for any $i$ and $j$. \\ 
$\Dg(x)$ & The diagonal matrix, whose diagonal elements \\
&are the elements of the vector $x$.\\
\hline
\end{tabular}
\end{table}
\begin{definition}[Social balance \cite{heider1967attitudes}]
An appraisal network $G(X)$ satisfies social balance if $X_{ii} > 0, \forall i \in \{1, \ldots, n\}$ and $\sgn(X_{ij}) \sgn(X_{jk}) \sgn(X_{ki}) = 1$, $\forall$ $i,j,k\in \{1, \ldots, n\}$.
\end{definition}

%The following lemma provides an equivalent condition for social balance.
\begin{lemma}[\hspace{-0.05em}\cite{mei2019dynamic}] 
	For any $X \in \mathbb{R}^{n \times n}$
	such that all of its entries are non-zero, $G(X)$ satisfies
	social balance if and only if it satisfies  $X_{ii} > 0$ for any $i \in \{1, \dots, n\}$ and
	$\sgn(X_{i*}) = \pm \sgn(X_{j*})$, for all $i$, $j \in \{1, \dots, n\}$.
\end{lemma}

\subsection{Model Description}
Our novel model of the interplay between homophily-based appraisal dynamics and influence-based opinion dynamics is formally presented as follows.  
\begin{definition}[The interplay model]
	Let $Y(t) \in \mathbb{R}^{n \times m}$ be the opinion matrix of $n$ agents towards $m$ issues at time $t$ and its entry $Y_{ij}(t)$ represents agent $i$'s opinion on issue $j$. Let $X(t)$ be the appraisal matrix at time $t$. Given the initial condition $Y(0)=Y_0\in\mathcal{Y}$, where the set $\mathcal{Y}= \{Y | Y(t) \in \mathcal{S}_{nz-row}, \text{ for any }t\geq 0,\text{ with }Y(0)=Y \}$, the interplay between the appraisal matrix $X(t)$ and the opinion matrix $Y(t)$ is given by the following dynamics:
	\begin{equation}
	\label{eq:X}
	X(t+1)=\Dg(|Y(t)|\mathbf{1}_{m})^{-1}Y(t)Y^{\top}(t),\\
	\end{equation}
	\begin{equation}
	\label{eq:W}
	W(t+1)=\Dg(|X(t+1)|\mathbf{1}_{n})^{-1}X(t+1),\\
	\end{equation}
	\begin{equation}
	\label{eq:Y}
	Y(t+1)=W(t+1)Y(t),
	\end{equation}
	or, equivalently, in the entry-wise form:
	\begin{equation}
	\label{eq:Xentry}
	X(t+1)_{ij}=\frac{\sum_{k=1}^{m} Y_{ik}(t)Y_{jk}(t)}{||Y_{i*}(t)||_{1}},\\
	\end{equation}
	\begin{equation}
	\label{eq:Wentry}
	W(t+1)_{ij}=\frac{ X_{ik}(t+1)}{||X^+_{i*}(t)||_{1}},\\
	\end{equation}
	\begin{equation}
	\label{eq:Yentry}
	Y(t+1)_{ij}=\sum_{k=1}^{n}W_{ik}(t+1)Y_{kj}(t). %=\frac{\sum_{k=1}^{n} X^+_{ik}(t+1)Y_{kj}(t)}{||X^+_{i*}(t+1)||_{1}}.
	\end{equation}
%where $\hma{X(t) = [X_{ij}(t)] \in \mathbb{R}^{n \times n}}$ is the appraisal matrix, which shows how agent $i$ thinks of agent $j$ at time $t$; $\hma{W(t)=[W_{ij}(t)]_{n\times n} \in \mathbb{R}^{n \times n}}$ is the influence matrix, which shows the influence agent $i$ assigns to agents $j$ at time $t$.
\end{definition}
The dynamical system~\eqref{eq:X}-\eqref{eq:Y} can be understood as the following iteration process: At each time $t+1$, via the person-to-opinion homophily mechanism, individuals form their interpersonal appraisals $X(t+1)$ based on the previous opinions $Y(t)$. Then a signed influence matrix $W(t+1)$ is constructed proportionally to the appraisal matrix $X(t+1)$. Finally, individuals update their opinions via the signed influence matrix $W(t+1)$, obeying the opinion dynamics model in~\cite{altafini2012consensus}. Note that $| \Dg(|Y(t)|1_m)^{-1}Y(t)|$ and $ |\Dg(|X(t+1)|1_n)^{-1}X(t+1)|$ are both row stochastic. This kind of normalization terms have been widely adopted in both appraisal dynamics \cite{mei2019dynamic} and opinion dynamics with antagonistic relations~\cite{hendrickx2014lifting,proskurnikov2017modulus,xia2015structural}. In addition, the assumption $Y(0)\in \mathcal{Y}$ implies that, starting with the initial condition $Y(0)$, at any time $t$, every agent $i$ has an non-zero opinion on at least one of the $m$ issues.
%System~\eqref{eq:X} is the homophily-based appraisal dynamics which describes the update law of the appraisal matrix $X(t)$ \cite{mei2019dynamic}. For any agent $i, j \in \{1, \dots, n\}$, agent $i$’s appraisal of agent $j$ at time $t + 1$ depends on to what extend they agree with each other on a set of issues. In particular, for any issue $k \in \{1, \dots, m\}$, if $\sgn(Y_{ik}(t)) = \sgn(Y_{jk}(t))$, the term $Y_{ik}(t)Y_{jk}(t)$ contributes positively to $X_{ij}(t + 1)$, and vice versa. The matrix $W(t)$ of equation~\eqref{eq:W} can be regarded as the influence matrix constructed from the appraisals through homophily mechanism. Equation \eqref{eq:Y} updates the opinion matrix based on influence-based opinion dynamics with negative weights \cite{altafini2012consensus}. This shows how opinion dynamics reacts to the evolution of influence network. \hma{The entry-wise absolute normalization terms $| \hma{\Dg}(|Y(t)|1_m)^{-1}Y(t)|$ and $ |\hma{\Dg}(|X(t+1)|1_n)^{-1}X(t+1)|$ are all row stochastic. This kind of normalization term has been widely adopted in both appraisal dynamics \cite{mei2019dynamic} and opinion dynamics with antagonism \cite{hendrickx2014lifting,proskurnikov2017modulus,xia2015structural}.}
	
By combining ~\eqref{eq:X}-\eqref{eq:Y}, we can obtain the following model in the form of the opinion dynamics:
\begin{equation} \label{eq:model2}
	Y(t+1)=\Dg(|Y(t)Y^{\top}(t)|\mathbf{1}_{n})^{-1}Y(t)Y^{\top}(t)Y(t).
\end{equation}
The system \eqref{eq:model2} concludes the interplay between homophily-based appraisal dynamics \eqref{eq:X} and influence-based opinion dynamics \eqref{eq:Y}. Obviously, $\mathcal{Y}$ is the domain of our model in~\eqref{eq:model2}. For the convenience of presentation, the time step $t$ can be omitted in case of no ambiguity, i.e., $X(t+1)$, $Y(t)$, and $Y(t+1)$ are denoted as $X^+$, $Y$, and $Y^+$, respectively.

\section{Theoretical Analysis}
In this section, we provide theoretical analysis regarding the interplay model~\eqref{eq:X}-\eqref{eq:Y}. The following lemma presents some finite-time properties of the interplay model.

\begin{lemma}[Finite-time behavior] 
	\label{Lm:lm2}
	Consider system~\eqref{eq:model2} and define $f(Y)=\Dg (|YY^{\top}|\mathbf{1}_{n})^{-1}YY^{\top}Y$. For any $Y_{0}\in \mathcal{Y}$, the following statements hold:
	
	\begin{enumerate}
		\item [i)] the map $f$ is well-defined for any $Y_{0}\in \mathcal{Y}$;
		
		\item [ii)]  the solution $Y(t)$, $t \in \mathbb{Z}_{\geq 0},$ to~\eqref{eq:model2} with the initial condition
		$Y(0) = Y_{0}$ exists and is unique;
		
		\item [iii)]  the max norm of $Y(t)$ satisfies
		$|Y(t+1)|_{\max} \leq |Y(t)|_{\max} \leq |Y(0)|_{\max};$
		
		\item [iv)]  for any $c> 0$, the trajectory $cY(t)$ is the solution to~\eqref{eq:model2}
		from initial condition $Y(0) = c Y_{0}$.
	
		\item [v)] $|X(t+1)|$ is upper bounded by $|Y(t)|_{\max}$ for any $t \in \mathbb{Z}_{\geq 0}$.
	\end{enumerate}
\end{lemma}

\begin{proof}
	According to the equation \eqref{eq:Wentry} and \eqref{eq:Yentry}, $Y^+=f(Y)$ is well-defined as long as $X^+\in \mathcal{S}_{nz-row}$, which is naturally true by dynamics~\eqref{eq:X}, since $Y(t) \in \mathcal{S}_{nz-row}$ for any $t \in \mathbb{Z}_{\geq 0}$ by assumption. 
	This fact leads to statement i).
	
	Statement ii) is the direct consequence of i).
	
	By the system equations~\eqref{eq:W} and~\eqref{eq:Y}, we have $|Y^+_{ij}|=\frac{|\sum_{k=1}^{n} X^+_{ik}Y_{kj}|}{||X^+_{i*}||_{1}} \leq \frac{\sum_{k=1}^{n} |X^+_{ik}||Y_{kj}|}{||X^+_{i*}||_{1}}\leq \max_{k}|Y_{kj}|\leq |Y(t)|_{\max}.$ It implies statement iii) immediately.
	
	Then, statement iv) is
	obtained by replacing $Y(t)$ with $cY(t)$ on the right-hand side
	of~\eqref{eq:model2}, 
	%\hma{i.e., let $\hat{Y}(t) = c Y(t)$ \dx{for any $t \geq 0$}. We %obtain 
%	\begin{equation*}
%	  \hat{Y}(t+1) = \Dg (|c^2Y Y^{\top}|\mathbf{1}_{n})^{-1}cY (cY)^{\top}cY=cY(t+1).
%	\end{equation*}
%    }
    i.e., let $\hat{Y}(0) = c Y_{0}$. We obtain 
	\begin{equation*}
	  \hat{Y}(1) = \Dg (|c^2Y_{0} Y_{0}^{\top}|\mathbf{1}_{n})^{-1}cY_{0} (cY_{0})^{\top}cY_{0}=cY(1).
	\end{equation*}
	For any $t \geq 1$, $\hat{Y}(t)=cY(t)$ holds true for the same calculation.

    Now we prove that statement v). According to equation~\eqref{eq:X},
	\begin{equation*}
	|X^+_{ij}|=\frac{\sum_{k=1}^{m}| Y_{ik}Y_{jk}|}{||Y_{i*}||_{1}} \leq \frac{\sum_{k=1}^{m} |Y_{ik}||Y_{jk}|}{||Y_{i*}||_{1}}\leq \max_{k}|Y_{jk}|.
	\end{equation*}
	It yields that $|X_{ij}(t+1)| \leq |Y(t)|_{\max}$ for any $t \in \mathbb{Z}_{\geq 0}$.
\end{proof}
\begin{remark}
 %\dx{The set $\{Y_{0}\in \mathcal{S}_{nz-row} | Y(t) \in \mathcal{S}_{nz-row},$ $t\geq 1 \}$ is the domain of our model in~\eqref{eq:model2}.}
 
 Specifically, in single-issue case ($m=1$), $\mathcal{S}_{nz-row}$ is an invariance set of the map $f$. For multi-issue case ($m>1$),  Monte Carlo validation indicates that, for any $Y$ randomly picked from $\mathcal{S}_{nz-row}$, $f(Y) \in \mathcal{S}_{nz-row}$ holds almost surely, see Section \ref{sec:val} for the simulation set-up and results.
\end{remark}
Before embarking on the main results on the equilibrium and the convergence of the opinion matrix $Y(t)$, we introduce the concept of \emph{non-vanishing} of the appraisal matrix $X(t)$.
\begin{definition}[\hspace{-0.06em}\cite{mei2019dynamic}] \label{Df:df1}
	A time-varying appraisal matrix $X(t)$ satisfies the non-vanishing appraisal condition if 
	${\liminf}_{t\rightarrow \infty}    {\min}_{i,j}|X_{ij}(t)|> 0$. 
\end{definition}

\begin{prop}
	\label{pro:pro1}
	For any initial condition $Y_{0}\in \mathcal{Y}$ such that $X(t)$ satisfies the non-vanishing appraisal condition, define the set $\Upsilon = \{Y(t)\}_{t=0}^{\infty}$. $\Upsilon$ is the invariance set of the map $f$. Moreover, $f$ is continuous on the set $\Upsilon$.
\end{prop}
\begin{proof}
    $\Upsilon$ is the invariance set of the map $f$ is equivalent to the statement: pick any $y \in \Upsilon$, $f(y)\in \Upsilon$ holds true.
	Since $X(t)$ satisfies the non-vanishing appraisal condition, $X(t)$ and $Y(t)$ must be defined at any time point $t > 0$. For any $Y(i) \in \Upsilon$, we have $Y^+=Y(i+1)\in \Upsilon$. Thus, $\Upsilon$ is the invariance set of the map $f$.
	
	Now, we show that a map is continuous at any isolated point. Let $X_{0}$ be any isolated point and $h$ be a map such that $Y_{0}=h(X_{0})$. For any $\epsilon>0$ and neighbor set of $Y_{0}$: $B_{Y_{0}}(\epsilon)$, there always exists $\delta>0$ and a neighbor set of $X_{0}$: $B_{X_{0}}(\delta)$ such that $B_{X_{0}}(\delta)=\{X_{0} \}.$ Thus, $f(B_{X_{0}}(\delta))=h(X_{0})=Y_{0}\in B_{Y_{0}}(\epsilon).$ 
	
	In addition, for any $Y(i) \in \Upsilon$, $Y(i)$ is the isolated point in the set $\Upsilon$, because there exists a neighborhood of $Y(i)$ which does not contain any other points of $\Upsilon$. By definition, a map $g : A \rightarrow B$ is continuous iff it is continuous at each point of $A$. Since any map is continuous at the isolated point and $\Upsilon$ contains only isolated points, $f$ is continuous on the set $\Upsilon$.
\end{proof}

\begin{thm}[Equilibrium set]
	\label{thm:fixed point multi}
	Given the opinion dynamics in~\eqref{eq:model2}, for any $Y^*\in \{Y | \quad |X^+|=|\Dg(|Y|\mathbf{1}_{m})^{-1}YY^{\top}|> \mathbf{0} \}$,
	 $Y^*$ is an equilibrium if and only if $Y^*$ is in the form of $Y^*=[a_{1}\rho,a_{2}\rho,...,a_{m}\rho]$, where $a_{i}\in \mathbb{R}$, $i=1, ...,m$, $\rho \in \{\pm 1\}^n$, and $\sum_{i=1}^{m} a_{i}^2\neq 0$.
\end{thm}
\begin{proof}
	By substituting $Y^*=[a_{1}\rho,a_{2}\rho,...,a_{m}\rho]$ into~\eqref{eq:model2}, we can directly obtain that $Y^+=Y^*$, which shows $Y^*=[a_{1}\rho,a_{2}\rho,...,a_{m}\rho]$ is the equilibrium of system of~\eqref{eq:model2}.
	
	Now, we prove $Y^*$ is the equilibrium only if $Y^*$ is in this form for any $Y^*\in \{Y | \quad |X^+|=|\Dg (|Y|\mathbf{1}_{m})^{-1}YY^{\top}|> \mathbf{0} \}$. 
	Let $i=\arg\max_{k} |Y_{kj}|,$ we know that,
	\begin{equation*}
	|Y^*_{ij}|=\frac{|\sum_{k=1}^{n} X^+_{ik}Y^*_{kj}|}{||X^+_{i*}||_{1}} \leq \frac{\sum_{k=1}^{n} |X^+_{ik}||Y^*_{kj}|}{||X^+_{i*}||_{1}}\leq \max_{k} |Y^*_{kj}|.
	\end{equation*}
	Notice that, all entries of $X$ are non-zero under non-vanishing appraisal condition and $X$ is obviously sign-symmetric. In order to make $|Y^*_{ij}|=\max_{k} |Y^*_{kj}|$ hold, the following conditions must be satisfied:
	\begin{enumerate}
		\item [a)] $\sgn(X^+_{i*})=\sgn(X^+_{*i})=\pm \sgn(Y^*_{*j})$;
		
		\item [b)] All entries of $Y^*_{*j}$ have magnitude $||Y_{*j}||_{\infty}$.
	\end{enumerate}
	It implies that for any $k$ and $l$, there exist $\delta_{k}, \delta_{l} \in \{\pm 1\}$ such that $\sgn(Y^*_{*k})=\delta_{k} \sgn(X_{i*})$ and $\sgn(Y^*_{*l})=\delta_{l} \sgn(X_{i*})$.
	Thus we have $\sgn(Y^*_{*k})=\pm \sgn(Y^*_{*l})$ and $|Y^*_{ik}|=||Y_{*k}^*||_{\infty}$. It is equivalent that the equilibrium can only be in the form $Y^*=[a_{1}\rho,a_{2}\rho,...,a_{m}\rho]$. Clearly, $Y^*$ is not allowed to be a zero matrix. Thus there holds $\sum_{i=1}^{m} a_{i}^2\neq 0$ and hence completes the proof.
\end{proof}
\begin{remark} 
	If the opinion matrix converges to a fixed point
	$Y^*=[a_{1}\rho,a_{2}\rho,...,a_{m}\rho]$, then the appraisal matrix $X(t)$ also converges to some equilibrium $X^*$ such that $X^*=\frac{\sum_{k=1}^{m} a_{k}^2}{\sum_{k=1}^{m} |a_{k}|}\rho\rho^\top$, which satisfies social balance. Obviously, any $Y^*=[a_{1}\rho,a_{2}\rho,...,a_{m}\rho]$ leads to a unique $X^*$, but not the other way around. That is, some $X^*$ at the steady state can be derived from different equilibria $Y^*$. 
\end{remark}
\begin{remark}
	Given an equilibrium $Y^*=[a_{1}\rho,a_{2}\rho,...,a_{m}\rho]$, there could exists some $i$ such that $a_{i}=0$. For example, the initial condition $Y(0)=\left[ 
	\begin{matrix}
	1 & 2 & 5 \\
	-1 & -2 & 5 \\
	-1 & -2 & 5 \\
	1 & 2 & 5 
	\end{matrix}
	\right]$ leads to the equilibrium 
	\begin{equation*}
	  Y^* = \left[
	  \begin{matrix}
	   0 & 0 & 5 \\
	   0 & 0 & 5 \\
	   0 & 0 & 5 \\
	   0 & 0 & 5 
	  \end{matrix}
	  \right] \mathrm{~and~}
	  X^* = \left[
	  \begin{matrix}
	   5& 5 & 5 & 5 \\
	   5& 5 & 5 & 5 \\
	   5& 5 & 5 & 5 \\
	   5& 5 & 5 & 5
	  \end{matrix}
	  \right].
	\end{equation*}
Specifically, in single-issue case, since all entries of the appraisal matrix $X(t)$ are obviously non-zero, $Y^*=a\rho$ where $a \neq 0$ is the unique equilibrium for any initial condition $Y(0) \in \mathcal{S}_{nz-row}$.
\end{remark}

\begin{remark}
For the case $Y^* \notin \{Y | \quad |X^+|=|\Dg (|Y|\mathbf{1}_{m})^{-1}YY^{\top}|> \mathbf{0}_{n \times n} \}$ , $Y^*$ is an equilibrium of the interplay model if $Y^*=P\hat{Y}P^\top$, where $P$ is a permutation matrix and $\hat{Y} =\Dg(A_{1},\dots,A_{l})$ with  $A_{k}=[a_{1}\rho,a_{2}\rho,...,a_{m}\rho]$ and $l\geq 2$. The corresponding $X^*=P\hat{X}P^\top$, where $P$ is a permutation matrix and $\hat{X}$ is a block diagonal matrix with blocks of the form $\alpha \rho \rho^\top$ and at least $2$ blocks. $X^*$ doesn't satisfy social balance and non-vanishing appraisal condition. However, for any generic initial condition, with $99\%$ confidence level, the system does not converge to this equilibrium. This will be validated in the section \ref{sec:val} by Monte Carlo validation.
\end{remark}

Before investigating the property of convergence, we provide the formal definition of modulus sign-consensus and \emph{modulus consensus}.
\begin{definition} (Modulus sign-consensus)\label{Df:mod_agr}
	The opinion dynamics reaches modulus sign-consensus, if there exists a time point $t_{0}\geq 0$ such that the non-zero opinion matrix $Y(t)$ satisfies $\sgn(Y(t)_{a*})=\pm \sgn(Y(t)_{b*})$ for any $a,b \in \{1,2,\ldots,n\}$ and $t\geq t_{0}$. Moreover, the opinion reaches \emph{sign-consensus}, if $\sgn(Y(t)_{a*})= \sgn(Y(t)_{b*})$ for any $a,b \in \{1,2,\ldots,n\}$; The opinion reaches \emph{bipartite sign-consensus}, if there exist certain $a,b \in \{1,2,\ldots,n\}$ such that $\sgn(Y(t)_{a*})= -\sgn(Y(t)_{b*})$.
\end{definition}

\begin{definition} (Modulus consensus) \label{Df:mod_consensus}
	The opinion dynamics reaches modulus consensus if there exists a time point $t_{0} \geq 0$ such that the non-zero opinion matrix $Y(t)$ satisfies $Y(t)_{a*}=\pm Y(t)_{b*}$ for any $a,b \in \{1,2,\ldots,n\}$ and $t \geq t_{0}$. Moreover, the opinion reaches consensus if $Y(t)_{a*}= Y(t)_{b*}$ for any $a,b \in \{1,2,\ldots,n\}$; The opinion reaches \emph{bipartite consensus} if there exist certain $a,b\in \{1,2,\ldots,n\}$ such that $Y(t)_{a*} = Y(t)_{b*}$.
\end{definition}
Definition \ref{Df:mod_agr} shows two possible cases for modulus sign-consensus. The first case is that there are only one faction among all agents, in which all the agents support or are against each issue. The other case is that there are two factions among all agents, i.e., there exist at least two agents $a$ and $b$ such that $\sgn(Y_{a*})=- \sgn(Y_{b*})$. The agents reach sign-consensus within each faction while they are against others across the factions. 
	
Modulus consensus in Definition \ref{Df:mod_consensus} is stricter than modulus sign-consensus since the magnitudes of opinions are taken into consideration. Note that modulus consensus in Definition \ref{Df:mod_consensus} has been already introduced in previous literature \cite{altafini2012consensus,anton_tac_2016} for single-issue case. We generalize the definition for the multi-issue case, which better describes agents' attitudes towards a set of issues. For our system, we have confirmed that the equilibrium of opinion matrix in Theorem \ref{thm:fixed point multi} satisfies modulus consensus, since the condition $Y^*_{a*}=\pm Y^*_{b*}$ for any $a$ and $b$ is apparently fulfilled. We then provide the following proposition to characterize the properties of the opinion dynamics.
\begin{prop}[Properties of modulus sign-consensus]	\label{pro:pro2}
	Consider the system in~\eqref{eq:model2}. For any given $t_{0}\geq 0$, 
	if $\sgn(Y_{a*}(t_{0}))=\pm
	\sgn(Y_{b*}(t_{0}))$ for any $a$ and $b$, the following statements hold true.
	\begin{enumerate}
		\item [i)] the opinion reaches modulus sign-consensus and there holds $\sgn(Y(t))=\sgn(Y(t_0))$ for all $t \geq t_0$.
		\item [ii)] for any $t\geq t_{0}\geq 0$,  $|Y(t)|_{\min}$ is non-decreasing and lower bounds $|X(t+1)|$.
	\end{enumerate}
\end{prop}
\begin{proof}
	By the model~\eqref{eq:model2}, there holds
	\begin{equation} \label{eq:prop_sgnY+}
	\begin{split}
	&\sgn (Y^+_{ij})=\sgn \left(\frac{\sum_{l=1}^{n}\sum_{k=1}^{m} Y_{ik}Y_{lk}Y_{lj}}{||X^+_{i*}||_{1}} \right)\\  
		& =\sgn \left(\frac{\sum_{k=1}^{n} X_{ik}Y_{kj}}{||X_{i*}||_{1}} \right) 
		=\sgn \left(\sum_{l=1}^{n}\sum_{k=1}^{m} Y_{ik}Y_{lk}Y_{lj} \right).
	\end{split}
	\end{equation}
	Firstly, we consider the case when $Y_{ij}(t_0) \neq 0$ for all $i \in \{1,2,\ldots,n\}, j \in \{1,2,\ldots,m\}$. We prove statement i) by induction. Clearly, there holds $\sgn(Y(t))=\sgn(Y(t_0))$ if $t = t_0$. Suppose $\sgn(Y(s))=\sgn(Y(t_0))$ for some $s > t_0$. It follows that $\sgn(Y_{a*}(s))=\pm \sgn(Y_{b*}(s))$ for any agents $a$ and $b$. Thus there exist $\delta_{il} = \pm 1$ such that 
	\begin{equation} \label{eq:prop_sgn_iklj}
	   \begin{aligned} 
	    &\sgn(Y_{ik}(s) Y_{lk}(s) Y_{lj}(s)) \\
	    &= \delta_{il} \sgn(Y_{lk}(s))^2 \sgn(Y_{lj}(s)) = \sgn(Y_{ij}(s)).
	   \end{aligned}
	\end{equation}
	In conjugation with~\eqref{eq:prop_sgnY+}, it follows that $\sgn(Y_{ij}(s +1))=\sgn(Y_{ij}(s))$ and hence $\sgn(Y(s +1))=\sgn(Y(s))$. By induction, one can obtain $\sgn(Y_{ij}(t+1))=\sgn(Y_{ij}(t_0))$ for all $t \geq t_0$, which is equivalent to statement i).
	
	Then if there exists certain $Y_{ap} = 0$, it yields that $Y_{*p}=\mathbf{0}_{n \times 1}$. Note that the equation~\eqref{eq:prop_sgn_iklj} still holds since its two sides both equal to zero. In this regard, statement i) is also true. 
	
	Now we prove statement ii). Since $\sgn(Y_{a*}(t))=\pm \sgn(Y_{b*}(t))$, we have
	\begin{equation*}
	\begin{split}
	|X^+_{ij}|&=\frac{| \sum_{k=1}^{m} Y_{ik}Y_{jk}|}{||Y_{i*}||_{1}} = \frac{ \sum_{k=1}^{m}|Y_{ik}||Y_{jk}|}{||Y_{i*}||_{1}}\geq \min_{k}|Y_{jk}|\\ &\geq |Y(t)|_{\min}.
	\end{split}
	\end{equation*}
	It follows that $\sgn(X_{il})=\sgn(Y_{i*}Y_{l*}^\top)=\delta_{il}$ and $\sgn(X_{ik})=\sgn(Y_{i*}Y_{k*}^\top)=\delta_{ik}$, where $\delta_{il},\delta_{ik} \in \{\pm 1\}$.
	
	Moreover, if all the entries of $Y_{*j}$ are non-zero, $\sgn(Y_{lj})=\delta_{il}\sgn(Y_{ij})$ and $\sgn(Y_{kj})=\delta_{ik}\sgn(Y_{ij})$. Since all the index of the above relationships are arbitrary, $\sgn(X_{i*})=\sgn(Y_{ij})\sgn(Y_{*j})=\pm \sgn(Y_{*j}),$ which implies 
	\begin{equation}
	\label{eq:eqsum}
	\left| \sum_{k=1}^{n} X^+_{ik}Y_{kj} \right|=\sum_{k=1}^{n} |X^+_{ik}||Y_{kj}|.
	\end{equation}
	Note that if all the entries of $Y_{*j}$ are zero, the equation~\eqref{eq:eqsum} still holds. Thus we can obtain
	\begin{equation*}
	\begin{split}
	|Y^+_{ij}|&=\frac{|\sum_{k=1}^{n} X^+_{ik}Y_{kj}|}{||X^+_{i*}||_{1}} = \frac{\sum_{k=1}^{n} |X^+_{ik}||Y_{kj}|}{||X^+_{i*}||_{1}}\\
	&\geq \min_{k}|Y_{kj}|\geq |Y(t)|_{\min}.
	\end{split}
	\end{equation*}
	This completes the proof of statement ii).
\end{proof}

Proposition~\ref{pro:pro2} implies that that the set of modulus sign-consensus opinion matrices is positively invariant along the dynamics~\eqref{eq:model2}. Moreover, if $Y(t_0)$ satisfies modulus sign-consensus, then $Y(t)\in \mathcal{S}_{nz-row}$ and $X(t)$ satisfies the non-vanishing condition for any $t\ge 0$. In addition, as Proposition~\ref{pro:pro2} implies, once the individuals form either one all-friendly faction or two antagonistic factions based on the signs of their opinions, they will stay in their factions from then on. 
%Proposition \ref{pro:pro2} manifests that the opinion reaches modulus sign-consensus once all the agents have either similar or opposite opinions for every issue that is considered. Moreover, their attitudes (positive or negative) towards every issue do not change from this moment on. In this regard, if the agents are classified according to their attitudes, they already forms one faction (if the opinion reaches sign-consensus) or two factions (if the opinion reaches bipartite sign-consensus). Hence, individuals will stay in their factions, which is similar with the concept of social balance. Inspired by this insight, we present the following theorem that characterizes the relation between social balance of the appraisal network and modulus consensus of the opinion dynamics.

The following theorem characterizes the relation between social balance of the appraisal network and modulus consensus of the opinions.
\begin{thm}[Social balance and modulus consensus]
	\label{thm:sta_mul}
	Given the model in~\eqref{eq:model2}, the following statements hold true
	\begin{enumerate}
		\item [a)] Any equilibrium in the form of $Y^*=[a_{1}\rho,a_{2}\rho,...,a_{m}\rho]$ where $a_{i}\neq 0, \forall i \in \{1,2,\ldots,n\}$ and $\rho \in \{ \pm 1\}^n$ is locally stable;  

		\item [b)] For any $Y_{0}\in \mathcal{Y}$, the following four statements are equivalent.
		\begin{enumerate}
			
			\item [i)] the solution of the appraisal matrix $X(t)$ satisfies the non-vanishing appraisal condition.
			
			\item [ii)] there exists $Y^*=[a_{1}\rho,a_{2}\rho,...,a_{m}\rho]$ such that $\lim_{t \rightarrow \infty}Y(t)=Y^*$ and $\lim_{t \rightarrow \infty}X(t)=X^*=\frac{\sum_{k=1}^{m} a_{k}^2}{\sum_{k=1}^{m} |a_{k}|}\rho\rho^\top,$ where $a\in \mathbb{R}$, $\rho \in \{\pm 1\}^n$, and $\sum_{i=1}^{n} a_{i}^2\neq 0$.
			
			\item [iii)] the opinion reaches modulus sign-consensus.
			
			\item [iv)] there exists $ t_{0} > 0$ such that $G(X(t+1))$ achieves social balance for any $t\geq t_{0} > 0$. 
		\end{enumerate}
	\end{enumerate}
\end{thm}

\begin{proof}
	For a), let $\Delta = (\Delta_{ij})_{n \times m}$, where $\max_{i}|\Delta_{ik}|=\Delta_{k}<|a_{k}|$ and $Y(0)= Y^*+\Delta$.
	For any $t\geq 0,$ we have $\sgn(Y(t))=\sgn(Y(0))=\sgn(Y^*)$ and $0\leq |a_{k}|-\Delta_{k}\leq \min_{i}|Y_{ik}(t)|\leq \max_{i}|Y_{ik}(t)|\leq |a_{k}|+\Delta_{k}$.
	It implies that for any $i$ and $j$, we have $Y_{ij}(t)=\alpha_{ij} \sgn(Y_{ij}^*)$, where $|a_{k}|-\Delta_{k}\leq \alpha_{ij} \leq |a_{k}|+\Delta_{k}$. It follows that
	\begin{equation*}
	\begin{split}
	|Y(t)-Y^*|_{\max}&=\max_{ij}|\alpha_{ij}\sgn(Y_{ij}^*)-|a|\sgn(Y_{ij}^*)|\\
	&= \max_{ij}|\alpha_{ij}-|a||\leq \Delta_{k}\leq \Delta=\max_{k}\Delta_{k}
	\end{split}
	\end{equation*}
	Accordingly, for any $\epsilon >0$, there exists $\omega=\min(\max_{k}|a_{k}|/2,\epsilon/2)$ such that for any $Y(0)$ satisfying $|Y(0)-Y^*|_{\max}<\omega$, there holds $|Y(t)-Y^*|_{\max}<\epsilon$ for any $t\geq 0$. By definition, this shows that $Y^*$ is locally stable if $a_{i}\neq 0$ for any $i$.
	
	Now we prove that statement b) is true. Firstly, we show that statement i) implies ii). By Proposition \ref{pro:pro1}, $f(x)$ is continuous on the set $\Upsilon = \{Y(t)\}_{t=0}^{\infty}$. $\Upsilon$ is already a compact set, since $0\leq |Y_{ij}(t)| \leq |Y(0)|_{\max}$. Define $V(Y_{*k})=||Y_{*k}||_{\infty}$. $V$ is continuous on $\Upsilon$ similar to the proof in Proposition \ref{pro:pro1}. By statement iii) in Lemma \ref{Lm:lm2}, $V(Y^+_{*k})-V(Y_{*k})\leq 0$ for any  $Y\in \Upsilon \subset \mathcal{S}_{nz-row}$. According to the extended LaSalle invariance principle in the article \cite{mei2017lasalle}, we can obtain that, for any $Y_{0}\in \Upsilon$, $Y(t)$ converges to the largest invariant set $\mathcal{M}$ of the set $\mathcal{T}=\{  Y \in \Upsilon |  V(Y^+_{*k})-V(Y_{*k})= 0 \}$.
	
	Now we characterize the largest invariant set $\mathcal{M}$. For any $Y\in \mathcal{M}$, $V(Y^+_{*j})=V(Y_{*j})=||Y_{*j}||_{\infty}.$ Suppose $|Y^+_{ij}|=\max_{l}|Y^+_{lj}|.$ We have, 
	\begin{equation} \label{eq:thm2_y+}
	 \begin{aligned}
	  |Y^+_{ij}|&=\frac{|\sum_{k=1}^{n} X^+_{ik}Y_{kj}|}{||X^+_{i*}||_{1}} \leq \frac{\sum_{k=1}^{n} |X^+_{ik}||Y_{kj}|}{||X^+_{i*}||_{1}} \\
	  &\leq  \max_{k}|Y_{kj}|=||Y_{*j}||_{\infty}.
	 \end{aligned}
	\end{equation}
	Since $X(t)$ satisfies the non-vanishing condition, it implies that ${\min}_{i,j}|X_{ij}(t)|> 0$ when $t\rightarrow \infty$. Hence, all entries of $X$ are non-zero. Besides, since all the inequalities in~\eqref{eq:thm2_y+} must hold as equalities, it requires that $\sgn(X_{i*}^+)=\sgn(X_{*i}^+)=\pm \sgn(Y_{*j})$ and all entries of $Y_{*j}$ must have magnitude $||Y_{*j}||_{\infty}$. Thus we can conclude that $\mathcal{M}=\{  Y = Y^*=[a_{1}\rho,a_{2}\rho,...,a_{m}\rho],$ where $0\leq a_{i}\leq||Y(0)_{*i}||_{\infty}$ and not all $a_{i}$ are 0, for any $i$. By inserting $Y^*$ directly into the equation \eqref{eq:X}, we can obtain that $\lim_{t \rightarrow \infty}X(t)=X^*=\frac{\sum_{k=1}^{m} a_{k}^2}{\sum_{k=1}^{m} |a_{k}|}\rho\rho^\top.$

	Then we illustrate that statement ii) implies iii). Since $\lim_{t \rightarrow \infty}Y(t)=Y^*$ and $\sgn(Y_{a*}(t))=\pm \sgn(Y_{b*}(t))$ for any $a$ and $b$ after some time point $t_{0} >0$,
	there exists a neighbor set $\mathcal{U}(Y^*)$ such that $\sgn(Y)=\sgn(Y^*)$ for any $Y\in \mathcal{U}(Y^*)$. It follows that opinion reaches modulus sign-consensus. 

Here we prove that statement iii) can imply iv). Since the opinion reaches modulus sign-consensus, there exists $ t_{0} > 0$ such that $\sgn(Y_{a*}(t))=\pm \sgn(Y_{b*}(t))$ for any $a$ , $b$ and $t\geq t_{0}$. By proposition \ref{pro:pro2}, it holds that $\sgn(Y(t))=\sgn(Y(t_0))$ for all $t\geq t_{0}$. We show $G(X(t+1))$ satisfies social balance for any $t\geq t_{0}> 0$ by definition. There exist $\delta_{ij},\delta_{jl}, \delta_{li} \in  \{ \pm 1\}$ such that $\sgn(X^+_{ij})=\sgn(\frac{ \sum_{k=1}^{m} Y_{ik}Y_{jk}}{||Y_{i*}||_{1}})=\delta_{ij}$, $\sgn(X^+_{jl})=\sgn(\frac{ \sum_{k=1}^{m} Y_{jk}Y_{lk}}{||Y_{j*}||_{1}})=\delta_{jl}$, and $\sgn(X^+_{li})=\sgn(\frac{ \sum_{k=1}^{m} Y_{ik}Y_{lk}}{||Y_{l*}||_{1}})=\delta_{li}$.	
	It yields that $\sgn(X^+_{ij})\sgn(X^+_{jl})\sgn(X^+_{li})=\delta_{ij}\delta_{jl}\delta_{li}$. By listing all possible results of the above product, we obtain $\delta_{ij}\delta_{jl}\delta_{li}=1$. Hence we have $\sgn(X^+_{ii})=\frac{ \sum_{k=1}^{m} Y_{ik}Y_{ik}}{||Y_{i}||_{1}}>0$. It confirms that $G(X(t+1))$ satisfies social balance for any $t\geq t_{0}\geq 0$.

It remains to show that statement iv) can imply i). Since $G(X(t+1))$ achieves social balance for any $t\geq t_{0} > 0$, there holds that $\sgn(X(t+1)_{ij}) \sgn(X(t+1)_{jk}) \sgn(X(t+1)_{ki}) = 1$ $\forall$ $i$, $j$, $k \in \{1, \ldots, n\}$ for $t\geq t_{0}$. It follows that $|\sgn(X(t+1)_{ij})|=1$ for any $i$ and $j$. Thus there must exists $\delta > 0$ such that $|X_{ij}(t+1)|> \delta$ for each $t \geq t_{0}$. Therefore, $  {\mathrm{inf}~}  {\min}_{i,j}|X_{ij}(t+1)|\geq \delta> 0$ for each $t\geq t_{0}$. This completes the proof.
\end{proof}
Theorem \ref{thm:sta_mul} reveals the equivalence among non-vanishing condition of the appraisal matrix, social balance of the appraisal network, and modulus consensus of the opinion dynamics. 
%From this point of view, our model provides a comprehensive way of understanding the formation mechanism of social balance and modulus consensus. 
Specifically, in single-issue case, the opinion matrix $Y(t)$ converges to $Y^*=\frac{||Y(0)||_{2}^2}{||Y(0)||_{1}}\sgn(Y(0))$ and the appraisal matrix $X(t)$ converges to $X^*=\frac{||Y(0)||_{2}^2}{||Y(0)||_{1}}\sgn(Y(0))\sgn(Y(0))^\top$ in one step for any $Y_{0} \in  \mathcal{S}_{nz-row}$. This can be proved by substituting the initial value directly into the system equations \eqref{eq:X}-\eqref{eq:Y}.
\begin{remark}
In this paper we do not introduce any assumptions on the logical connections between different issues. However, as a result of the interplay between appraisals and opinions, we find that individuals' opinions on different issues are eventually correlated with each other, and their opinions, are determined by which faction they are in. This feature reflects the so-called \emph{opinion partisanship} phenomenon in political science~\cite{DB-AG:08}, i.e., the correlation of issue attitudes with party identification.
\end{remark}

\section{Numerical Examples And Validation}
In this section, we illustrate our main results by numerical simulations. Besides, Monte Carlo simulations are conducted to 
support the validity of the non-vanishing condition. 
\subsection{Numerical Examples} \label{sec:num}
We use here the generic initial condition. By generic initial condition, it means that each entry of $Y_{0} \in  \mathcal{S}_{nz-row}$ is independently randomly generated from the uniform distribution on some support $[-a,a]$. According to the Lemma \ref{Lm:lm2}, our model is independent of scaling. In this regard, we decide to set $a=1$.

We characterize the general case with 9 agents and 6 issues. See Fig.~\ref{fig:vis} for the visualization of the evolution of the appraisal matrix and the opinion matrix, respectively. It can be observed that in $9$ steps, the appraisal matrix $X(t)$ achieves social balance and the opinion matrix $Y(t)$ reaches bipartite consensus and opinion partisanship. 
\begin{figure}
\centering
\includegraphics[width=0.99\linewidth]{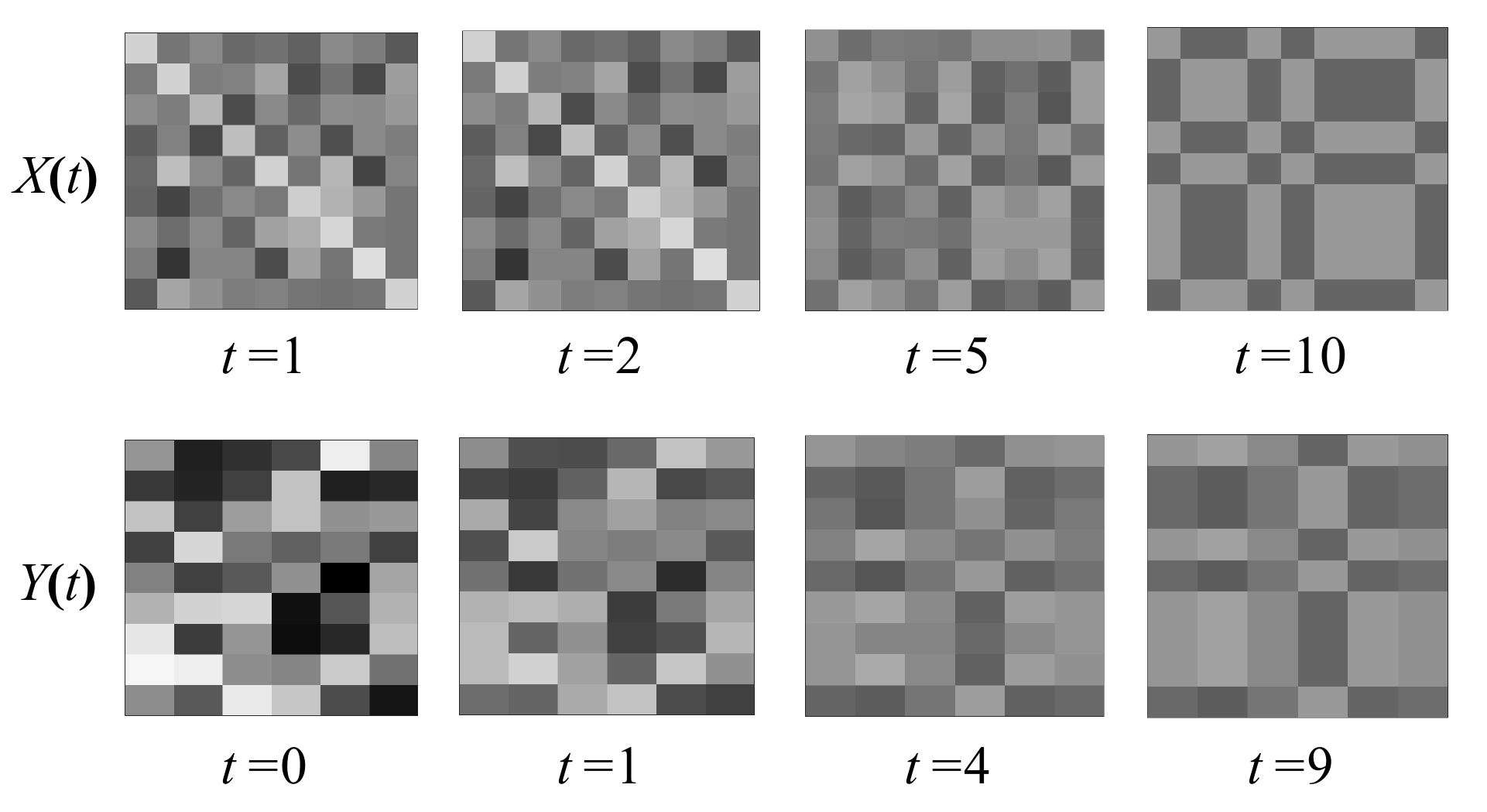}
\caption{Visualizations of the emergence of social balance and opinion partisanship in a group of 9 individuals discussing on 6 issues. In the visualized matrices, each grey block is one entry and the lower the entries value, the darker its color is. The first row of this figure shows that the appraisal network evolves to two antagonistic factions. The second row of this figure shows that, together with the appraisal dynamics, the individuals' opinions on different issues are eventually correlated with each other and dependent purely on what factions they are in.}
\label{fig:vis}
\end{figure}

\subsection{Numerical Validation of the Non-Vanishing Appraisal Condition} \label{sec:val}
We conduct Monte Carlo simulations to validate the non-vanishing appraisal condition for generic initial condition. The generic initial conditions are constructed by using the same technique in Section \ref{sec:num}. For any randomly generated generic initial condition $Y(0)$, the random variable $Z(Y(0))$ is defined by
\begin{equation}
Z(Y(0))=\left\{
\begin{array}{rcl}
1, & & {\text{if} \min\limits_{100\leq t \leq 1000}\min_{i,j}|X_{ij}(t)|\geq 0.001}\\
0, & & {\text{otherwise}} 
\end{array} \right.
\end{equation}
By running the simulation $N$ times independently, we can obtain samples $Z_1, Z_2, \ldots, Z_N$. It follows that the frequency of the occurrences of non-vanishing appraisal matrix is $\hat{p}_{N}=\sum_{i=1}^{N}Z_{i}/N$. Then the frequency $\hat{p}$ can be utilized to approximate the probability $p=\mathsf{Prob}(Z(Y(0))=1)$ with large $N$. Let $1-\xi \in (0,1)$ and $1-\xi \in (0,1)$ be accuracy and confidence level. The approximation error, i.e., $|\hat{p}_{N}-p|$, is bounded by $\epsilon$ with probability greater than $1-\xi$, if the Chernoff bound is fulfilled: $N \geq \frac{1}{2\epsilon^2}log\frac{2}{\xi}$. In this paper, by setting $\epsilon=\xi=0.01$ and conducting $N=27000$ simulations with 9 agents and 4 issues, we find that $\hat{p}=1$. In this context, we can conclude that, for any generic initial condition with $99\%$ confidence level, there is at least $99\%$ probability that every entry of $|X(t)|$ is greater than a small positive scalar ($0.001$), for all $t\in [100,1000]$.
Notice that $Y(t) \in \mathcal{S}_{nz-row}$ for $t=[1,999]$ holds if $X(1000)$ exists. Therefore, this simulation also validates that for any $Y$ randomly picked from $\mathcal{S}_{nz-row}$, $f(Y) \in \mathcal{S}_{nz-row}$ holds almost surely. 

\section{Conclusion}
In this paper, we propose a new discrete-time nonlinear model that characterizes the interplay between appraisal dynamics and opinion dynamics. In particular, the appraisal dynamics is based on person-entity homophily and the opinion dynamics is an influence-based Altafini-like model on coopetitive networks. Based on theoretical analysis and numerical experiments, the evolution of this interplay model show that the appraisal matrix reaches social balance and the opinion matrix achieves modulus consensus for almost any initial condition.

\bibliographystyle{IEEEtran}
\bibliography{bib}

\end{document}